\documentclass[10pt, a4paper]{article}
\usepackage{graphicx} 
\usepackage{amsfonts}
\usepackage{amsmath}
\usepackage{amssymb}
\usepackage{amsthm}
\usepackage[T1]{fontenc}
\usepackage{framed}

\newcommand{\bra}[1]{\left\langle{#1}\right|}
\newcommand{\ket}[1]{\left|{#1}\right\rangle}
\newcommand{\braket}[2]{\left\langle{#1}|{#2}\right\rangle}
\newcommand{\ketbra}[2]{|{#1}\rangle\langle{#2}|}
\newcommand{\braOPket}[3]{\left\langle{#1}|{#2}|{#3}\right\rangle}
\newcommand{\Tr}[0]{\operatorname{Tr}}
\newcommand{\tr}[0]{\operatorname{tr}}
\newcommand{\Mat}[0]{\mathbf{Mat}_{\mathbb{R}_{\geq 0}}}
\newcommand{\id}[0]{\mathrm{id}}
\newcommand{\HC}[0]{\mathcal{H}_C}
\newcommand{\HCC}[0]{\mathcal{H}_{\widetilde{C}}}
\newcommand{\HCin}[0]{\mathcal{H}_{C_{\mathrm{in}}}}
\newcommand{\HCout}[0]{\mathcal{H}_{C_{\mathrm{out}}}}

\newtheorem{theorem}{Theorem}
\newtheorem{definition}{Definition}
\newtheorem{corollary}{Corollary}
\newtheorem{lemma}{Lemma}

\title{Tracing the Loop: \\ \large Non-Causal Computation, Partial Traces, \& Postselected Entanglement}
\author{Mark Carney \\ \emph{Quantum Village Inc.}}
\date{September 2026}

\begin{document}

\maketitle

\begin{abstract}
    This paper gives a categorical interpretation of Baumeler \& Wolf's logically consistent non-causal circuits, connecting them to postselected quantum teleportation. Looped feedback is represented by a trace in the category of non-negative matrices, and it is shown that the traced process is stochastic precisely when the induced loop transition matrix has trace $1$ for every external input, a condition shown to be equivalent to a unique fixed point for the loop for each input to a deterministic circuit. A classical non-causal circuit is represented by a measure-and-prepare quantum channel with an internal register utilising a maximally entangled Bell-state with postselection. The main result is that classical logical consistency is equivalent to the postselected Bell outcome having, for loop dimension $d$, probability exactly $1/d^2$ for each classical input distribution. The subsequent normalised conditional output then agrees exactly with the classical categorical trace. This work identifies a class of quantum Bell-postselection constructions whose conditional evolution maintains linear dependence on classical input distributions.
\end{abstract}

\section{Introduction and Preliminaries}

This work gives a categorical interpretation of Baumeler \& Wolf's non-causal circuits \cite{BaumelerWolf2017} relative to categorical quantum mechanics and Postselected Closed Timelike Curves (P-CTCs) from Lloyd \emph{et al.} \cite{Lloyd2011-zc}. The main result may be stated as follows:

\begin{center}
\begin{framed}
    \noindent A non-causal wire is a trace, with entanglement being the quantum realization through cup/cap decomposition.
\end{framed}
\end{center}

To that end, this work builds on Ara\'ujo \emph{et al.} \cite{Arajo2017}, who first identified a correspondence between quantum process matrices and a linear sub-class of P-CTCs. Ara\'ujo \emph{et al.}'s setting involves checking validity under arbitrary local operations, whilst in this work we consider an examination of a specified classical circuit across external inputs and outputs. Indeed, the technique used below, that of constructing a partial trace by means of restriction of an ambient trace, is itself taken from Malherbe, Scott, \& Selinger \cite{Malherbe2012}. The contribution from this work is an explicit categorical formulation of classical feedback admissibility, relating it to a measure-and-prepare embedding through Bell-state preparation and postselection.

\subsection{Non-Causal Computation}

The main starting point for this work is Baumeler \& Wolf \cite{BaumelerWolf2017}. Their paper proposes a model of non-causal computation that characterizes logical consistency of a non-causal computation by the existence of a unique fixed point. By \cite[\S II]{BaumelerWolf2017}, let $f: C \to C$ be a Markov chain obtained by cutting the `loop' in a non-causal circuit, the circuit is logically consistent if and only if $$ \exists!x\quad f(x) = x $$

Similarly, for 0-1 Markov matrix $\hat{H}$, a closed circuit is consistent if and only if for some $e_c$, a vector with one $1$ in the $c^{th}$ postion and $0$'s elsewhere, $e_c = \hat{H}e_c$ for some unique $c$ - specifically, the diagonal of $\hat{H}$ has a unique entry of $1$.

Likewise, an open circuit with input $a$, output $x$, and looping wire $c$ is consistent if and only if for every $a \in A$ there is a unique $(x, c) \in X \times C$ such that \cite[\S III]{BaumelerWolf2017} $$ (e_x \otimes e_c)^T \hat{H} (e_a \otimes e_c) = 1 $$ 

In the stocastic case, their earlier paper \cite{Baumeler2016-so} shows that consistency arises when the average number of fixed points is one, read as $\tr(\hat{H}) = 1$.

\subsection{Category Theory}

To better understand this work, it is best to think that category theory posits the algebra of composition, and a category thereby abstracts the notion of a process. For this work, a good intuition is that a category contains objects and composable arrows between them, with the only structure being sequential composition. Monoidal categories add parallel composition within a process, and traced monoidal categories add feedback to these processes.

This work makes use of traced monoidal categories, in particular the fact that every compact closed category carries a canonical trace \cite{Joyal1996-bg}. We will relate this to string diagrams and their cup/cap semantics, utilising the categorical treatment of quantum mechanics by Abramsky \& Coecke \cite{Abramsky2004-xp} and Coecke \& Kissinger \cite{Coecke2017-kx}. Relevant results will be cited along the way. 

Properties of the Markov category $\mathbf{Stoch}$ are found in Fritz \cite{Fritz2020-dd} and Cho \& Jacobs \cite{ChoJacobs2019}.

\subsection{Postselected Closed Timelike Curves}

G\"{o}del \cite{Gdel1949} exhibited a cosmological solution to Einstein's equations containing closed timelike curves (CTCs), and CTCs received their first quantum mechanical treatment by Deutsch \cite{Deutsch1991}. This work uses the presentation of postselected closed timelike curves (P-CTCs) in papers by Lloyd \emph{et al.}, specifically \cite{Lloyd2011b, Lloyd2011-zc}. A P-CTC emulates a CTC by means of quantum teleportation, crucially without the classical correction step. 

A closed timelike curve system $C$ is essentially a Hilbert space equipped with a closed loop - this is when one half of a maximally entangled pair is fed into the ``past'' initial state of a system end of the loop. When the system evolves, the ``future'' end of the loop is then projected, together with the other member of the entangled pair, onto the same maximally entangled state. This works in postselection by keeping only the runs in which this projection succeeds. The net effect is that the loop is closed. That is, the state emerging from the future end is identified with the state that entered at the initial end. In string diagram terms, the entangled pair is a cup and the projection is a cap, so a P-CTC is then just a wire bent round to meet itself - fitting the intuition of a `closed loop' very well. 

The probability of the projection succeeding is inversely proportional to the square of the dimension of the looped system. In a classical setting, the projection vanishes exactly for any ``paradoxical'' evolutions with no consistent assignment to the loop. However, under quantum dynamics zero probability may also arise from interference. Deutsch's model \cite{Deutsch1991}, fixes the state on the loop by a self-consistency condition on density matrices. Lloyd \emph{et al.}'s P-CTC map is linear on the input state for unnormalized maps, however normalised evolution for successful postselection requires conditioning that divides the probability of success, and so is non-linear. Connecting this to Baumeler and Wolf \cite[§VI]{BaumelerWolf2017} therefore requires some restrictions on the P-CTC construction.

Later, Section \ref{sec:Bell} will show that Baumeler and Wolf's classical model is the measure-and-prepare special case of a P-CTC, with their logical-consistency condition equivalent to a Bell projection success with probability $1/d^2$, for $d$ the loop dimension. Normalization multiplies the output by $d^2$ and so the conditional evolution is just the stochastic map found by tracing the (classical) loop wire. This then retains linear dependence on input distributions that in \cite[\S 6]{BaumelerWolf2017} distinguishes their model from Deutsch.

\section{Classical Non-Causal Feedback as Partial Categorical Traces}

Let $M(b \mid a)$ denote the matrix entry\footnote{This is the same as the regular notation $M_{ba}$, except it allows emphasis on input-output interpretations.} of $M$ with entries in $\mathbb{R}_{\geq 0}$ that takes input $a$ to output $b$. For a column-stochastic matrix, $$ M(b \mid a) = \Pr(b | a) $$ Each input $a$ gives a column of probabilities over $b$ with $$ \sum_b M(b \mid a) = 1 $$

\begin{definition}
    Let $\mathbf{Mat}_{\mathbb{R}_{\geq 0}}$ denote the category whose objects are finite sets and whose morphisms $M : A \rightarrow B$ are matrices with entries $M_{ba} \in \mathbb{R}_{\geq 0}$. 
    
    Let \emph{$\mathbf{Stoch} \subseteq \mathbf{Mat}_{\mathbb{R}_{\geq 0}}$} be the symmetric monoidal subcategory of column-stochastic matrices.
\end{definition}

Composition in $\Mat$ is just matrix multiplication, with $$ (N \circ M)(c \mid a) = \sum_b N(c \mid b)M(b \mid a) $$ Monoidal products are just Cartesian products, $A \otimes B = A \times B$, and the tensor product of morphisms are just the standard Kronecker product $$ (M \otimes N)(b,d \mid a,c) = M(b \mid a)N(d \mid c) $$ The monoidal unit is the singleton set $I = \{ * \}$.

For a finite set $C$ define the self-dual object $C^* = C$, and let the $\eta_C$ coevaluation morphism and $\epsilon_C$ evaluation morphism be defined by:
\begin{align*}
    \eta_C : I \rightarrow C \otimes C, \quad \eta_C(c, c' \mid *) = \delta_{c,c'} \\
    \epsilon_C : C \otimes C \rightarrow I, \quad \epsilon_C(* \mid c,c') = \delta_{c,c'}
\end{align*} where $\delta_{c,c'}$ is the Kronecker delta function, returning 1 if $c = c'$, 0 if $c \neq c'$. $\eta_C$ represents cup, and $\epsilon_C$ represents cap semantics, satisfying the snake equations. This structure gives $\mathbf{Mat}_{\mathbb{R}_{\geq 0}}$ a compact structure as a matrix category over a commutative semi-ring (see \cite{Coecke2017-kx}). 

Clearly $\eta_C$ and $\epsilon_C$ are morphisms of $\mathbf{Mat}_{\mathbb{R}_{\geq 0}}$
but for dimension $d > 1$ they are not morphisms of $\mathbf{Stoch}$; the column of $\eta_C$ sums to $d = |C|$, whilst a column of $\epsilon_C$ indexed by some $(c, c')$ for $c, c' \in C$ sums to $\delta_{c, c'}$, which is zero whenever $c \neq c'$. $\mathbf{Stoch}$ does not inherit the compactness of $\mathbf{Mat}_{\mathbb{R}_{\geq 0}}$.

Let a morphism $P$ be 
\begin{align}\label{eq:morphismP}
    P: A \times C \rightarrow X \times C
\end{align}
Its matrix entries are given by $P(x, c_{\mathrm{out}} \mid a, c_{\mathrm{in}})$. Here, $a$ and $x$ are external input/output respectively, and $c_{\mathrm{in}}, c_{\mathrm{out}}$ are internal input/output within $C$.

\begin{definition}[Categorical Trace]\label{def:catTrace}
The categorical trace of $P$, as defined in (\ref{eq:morphismP}), is given by $$ \Tr_C(P) : A \rightarrow X $$ defined as $$ \Tr_C(P) = (\id_X \otimes \epsilon_C) \circ (P \otimes \id_C) \circ (\id_A \otimes \eta_C) $$ \end{definition}

We may calculate the matrix entries for $\Tr_C$ given by 

\begin{align} 
    \Tr_C(P)(x \mid a ) = \sum_{c, c' \in C} P(x, c' \mid a, c)\delta_{c,c'} = \sum_{c \in C} P(x, c \mid a, c)
\end{align}

$\Tr_C(P)$ is the induced partial trace on $\mathbf{Stoch}$, and has value only if $P$ is stochastic (and is undefiend otherwise). This partial trace comes from the inclusion $\mathrm{Stoch} \subseteq \mathrm{Mat}_{\mathbb R_{\ge0}}$, as found in \cite[Prop. 3.20]{Malherbe2012} Given that $\eta_C$ and $\epsilon_C$ lie outside $\mathbf{Stoch}$, $\mathrm{Tr}_C(P)$ is not necessarily stochastic even when $P$ is stochastic. Theorem \ref{thm:stoch-is-Tr} below characterizes exactly when it arrives into $\mathbf{Stoch}$. This is the categorical treatment of logical consistency defined in \cite{BaumelerWolf2017}.

\subsection{Logical Consistency is Trace Admissibility}

\begin{lemma}\label{lemma:morphismP}
Let $P$ be a stochastic matrix as per eq. (\ref{eq:morphismP}). Then for every pair $(a, c_{\mathrm{in}}) \in A \times C$ $$ \sum_{x \in X}\sum_{c_{\mathrm{out}}\in C} P(x, c_{\mathrm{out}} \mid a, c_{\mathrm{in}}) = 1 $$ 
\end{lemma}

\begin{proof}
    For $(a, c_{\mathrm{in}})\in A \times C$ it suffices to rearrange the sum for a column-stochastic matrix $$ \sum_{(x,c_{\mathrm{out}}) \in X \times C} P(x, c_{\mathrm{out}} \mid a, c_{\mathrm{in}}) = 1 $$ Note that for each summand, $a$ and $c_{\mathrm{in}}$ are fixed.
\end{proof}

For each external input $a$, define  $H_a : C \rightarrow C $ as $$ H_a (c_{\mathrm{out}} \mid c_{\mathrm{in}}) = \sum_{x \in X} P(x, c_{\mathrm{out}} \mid a, c_{\mathrm{in}}) $$ $H_a$ is stochastic, since $$ \sum_{c_{\mathrm{out}} \in C} H_a(c_{\mathrm{out}} \mid c_{\mathrm{in}}) = \sum_{c_{\mathrm{out}}\in C}\sum_{x \in X} P(x, c_{\mathrm{out}} \mid a, c_{\mathrm{in}}) = 1 $$ 

\begin{theorem}\label{thm:stoch-is-Tr}
    The following are equivalent: 
    \begin{enumerate}
        \item $\Tr_C(P) : A \rightarrow X$ is stochastic, \emph{i.e.} $\Tr_C(P) \in \mathbf{Stoch}$,
        \item For each $a \in A$, $\tr(H_a) = 1$,
        \item For each $a \in A$, $$ \sum_{x \in X} \sum_{c \in C}P(x, c\mid a, c) = 1 $$
    \end{enumerate}
\end{theorem}

Condition 2 in Theorem \ref{thm:stoch-is-Tr} is exactly per \cite[eq. (2)]{BaumelerWolf2017}. Theorem \ref{thm:stoch-is-Tr} thus shows that Baumeler-Wolf's trace condition is precisely the categorical trace defined above going into $\mathbf{Stoch}$.

\begin{proof}
    Note that 
    \begin{align*} 
        \sum_{x \in X} \Tr_C(P) (x \mid a) &= \sum_{x \in X} \sum_{c \in C} P(x, c \mid a, c) \\
        &= \sum_{c \in C } H_a (c \mid c) \\
        &= \tr (H_a)
    \end{align*} 
    Given $\Tr_C(P)$ is non-negative, it is stochastic if and only if the quantity above $=1$ for every $a$. 

\end{proof}

\begin{theorem}
    Let $P$ be deterministic with the same definition as eq. (\ref{eq:morphismP}), so that there exist component functions $h: A \times C \to X$ and $g: A \times C \to C$ such that $$ P(x, c_{\mathrm{out}} \mid a, c_{\mathrm{in}}) = \delta_{x, h(a, c_{\mathrm{in}})} \delta_{c_{\mathrm{out}}, g(a, c_{\mathrm{in}})} $$ The following are equivalent:
    \begin{enumerate}
        \item $\Tr_C(P) \in \mathbf{Stoch}$
        \item For each $a$ there exists a unique $c_a$ such that $g(a, c_a) = c_a$.
    \end{enumerate}
\end{theorem}

\begin{proof}
    Substituting the categorical trace formula we get 
    \begin{align*}
        \Tr_C(P)(x \mid a) &= \sum_{c \in C} P(x, c \mid a, c) \\
        &= \sum_{c \in C} \delta_{x, h(a, c)} \delta_{c, g(a, c)}
    \end{align*}

    The factor $\delta_{c, g(a, c)} = 1$ precisely when $g(a,c) = c$, therefore $$ \Tr_C(P)(x \mid a) = |F| $$ where $$ F = \{ c \in C : g(a,c)=c \text{ and } h(a,c) = x \}$$ Summing over $x$:
    \begin{align*}
        \sum_{x \in X} \Tr_C(P)(x\mid a) &= \sum_{x \in X} \sum_{c \in C} \delta_{x, h(a, c)} \delta_{c, g(a, c)} \\
        &= \sum_{c \in C} \delta_{c, g(a,c)} \sum_{x\in X} \delta_{x, h(a, c)} \\
        &= \sum_{c \in C} \delta_{c, g(a,c)}
    \end{align*}
    The last step following from the fact that $\sum_{x\in X} \delta_{x, h(a, c)} = 1$ for every fixed $(a,c)$. The right hand side is precisely the fixed points of $g$, and so $\Tr_C(P) \in \mathbf{Stoch}$ if and only if $g$ has exactly one fixed point for each $a$. This completes the equivalence.
    
    Let $c_a$ be the fixed point for some $a$ such that $g(a,c_a) = c_a$, then every term in the sum $$ \sum_{x \in X} \Tr_C(P)(x\mid a) = \sum_{x \in X} \sum_{c \in C} \delta_{x, h(a, c)} \delta_{c, g(a, c)} $$ goes to zero except when $c = c_a$. So the sum reduces to the single contribution from $$ \Tr_C(P)(x \mid a) = \delta_{x, h(a, c_a)} $$ Thus the traced process is the deterministic function $$ a \mapsto h(a, c_a)$$
\end{proof}

\section{Bell-State Realizations}\label{sec:Bell}

Let $C$ be a finite set as before, with $|C|=d$. Let $\HC$ be a $d$-dimensional Hilbert space with orthonormal basis given by $\{ \ket{c} : c \in C \}$, and let $\HCC$ be a copy of $\HC$. We will need to differentiate between different input and output physical tensor factors, so let these be denoted by $$ \HCin \cong \HCout \cong \HCC $$ Let $P$ be stochastic such that $P : A \times C \to X \times C$ as before, and assume nonempty $A$, $X$, $C \neq \emptyset$. 

Let $\mathcal{L}(\mathcal{H})$ denote the linear transformations of a Hilbert space. Associate a representation of $P$ to the quantum channel $$ \mathcal{E}_P : \mathcal{L}(\mathcal{H}_A \otimes \HCin) \to \mathcal{L}(\mathcal{H}_X \otimes \HCout) $$ defined for some quantum state $\rho$ by
\begin{equation}\label{eq:E_P}
    \mathcal{E}_P(\rho) = \sum_{\substack{a \in A , \\ c_{\mathrm{in}}, c_{\mathrm{out}} \in C, \\ x \in X}} P(x,c_{\mathrm{out}} \mid a,c_{\mathrm{in}}) \bra{a,c_{\mathrm{in}}} \rho \ket{a,c_{\mathrm{in}}} \ket{x, c_{\mathrm{out}}} \bra{ x,c_{\mathrm{out}} }
\end{equation} This assumes the natural `measure-and-prepare' embedding of a stochastic matrix into a quantum channel, for example see \cite{horodecki2003}.

Let the external input be the classical state from input set $A$ given by 
\begin{equation}\label{eq:rhoA}
    \rho_A = \sum_{a \in A} p(a) \ket{a}\bra{a}
\end{equation} Consider the following procedure: 
\begin{enumerate}
    \item Prepare the state $|\Phi\rangle_{C_{\mathrm{in}}\widetilde{C}}$,
    \item Apply $\mathcal{E}_P$ to systems $A \otimes C_{\mathrm{in}}$ producing a state on $X \otimes C_{\mathrm{out}}$, 
    \item Project $C_{\mathrm{out}} \otimes \widetilde{C}$ into $\ket{\Phi}_{C_{\mathrm{out}} \widetilde{C}}$
    \item Retain the resulting unnormalised state on $X$ as $\sigma_X$.
\end{enumerate}

We may describe this retained state by $$ \sigma_X = \frac{1}{d^2} \sum_{a,x} p(a) \Tr_C(P)(x \mid a) \ket{x}\bra{x} $$ and so by substitution we may obtain $$ p_\Phi = \frac{1}{d^2} \sum_{a \in A} p(a) \sum_{x \in X} \Tr_C(P)(x \mid a) $$

\begin{theorem}\label{thm:quantumTFAE}
    The following are equivalent:
    \begin{enumerate}
        \item $\Tr_C(P) \in \mathbf{Stoch}$
        \item For each $a \in A$, $\tr(H_a) = 1$,
        \item For each classical probability distribution $p$ on A $$ p_\Phi = \frac{1}{d^2} $$
    \end{enumerate}
\end{theorem}

\begin{proof}
    The initial joint state is given by $\rho_A \otimes \ket{\Phi_{\mathrm{in}}}\bra{
    \Phi_{\mathrm{in}}}$ where $$ \ket{\Phi_{\mathrm{in}}} = \frac{1}{\sqrt{d}} \sum_{c \in C} \ket{c}_{C_{\mathrm{in}}} \otimes \ket{c}_{\widetilde{C}} $$ and so $$ \ket{\Phi_{\mathrm{in}}}\bra{\Phi_{\mathrm{in}}} = \frac{1}{d} \sum_{c, c' \in C} \ket{c}\bra{c'}_{C_{\mathrm{in}}} \otimes \ket{c}\bra{c'}_{\widetilde{C}} $$ 
    
    The joint input state, tensoring with eq. (\ref{eq:rhoA}), is then $$ \frac{1}{d} \sum_{a,c,\tilde{c}} p(a) \ket{a,c}\bra{a,\tilde{c}}_{A C_{\mathrm{in}}} \otimes \ket{c}\bra{\tilde{c}}_{\widetilde C} $$ Recalling the definition of $\mathcal{E}_P$ from eq. (\ref{eq:E_P}), it can be noted that by applying $\mathcal{E}_P \otimes \id_{\widetilde{C}} $, only matrix elements diagonal in the basis $ \{\ket{a, c_{\mathrm{in}}}\}$ contribute terms. In fact, $$ \bra{a_0, c_0} (\ket{a, c} \bra{a, \tilde{c}} ) \ket{a_0, c_0} = \delta_{a_0, a} \delta_{c_0, c} \delta_{c_0, \tilde{c}} $$ As such, the contribution is zero unless $c = \tilde{c}$. 

    After applying $\mathcal{E}_P \otimes \id_{\widetilde{C}} $, the resulting state is given by $$ \omega_{X C_{\mathrm{out}}\widetilde C} = \frac{1}{d} \sum_{a,c,x,c'} p(a) P(x,c'\mid a,c) \ket{x,c'} \bra{x, c'}_{X C_{\mathrm{out}}} \otimes \ket{c}\bra{c}_{\widetilde C} $$ with $c$ labeling the basis in $C_{\mathrm{in}}$ and $c'$ labeling the basis in $C_{\mathrm{out}}$.

    Now define $$ \ket{\Phi_{\mathrm{out}}} = \frac{1}{\sqrt{d}} \sum_{k \in C} \ket{k}_{C_{\mathrm{out}}} \otimes \ket{k}_{\widetilde{C}} $$ and let the projector $\Pi_{\mathrm{out}} = \ket{\Phi_{\mathrm{out}}} \bra{\Phi_{\mathrm{out}}}$.

    Next, postselect on this Bell outcome on $C_{\mathrm{out}} \otimes \widetilde{C}$, then the final unnormalised state on $X$ is given by
    \begin{equation}\label{eq:sigmaXQuantum}
        \sigma_X = \Tr_{C_{\mathrm{out}}\widetilde{C}} [(I_X \otimes \Pi_{\mathrm{out}}) \omega_{X C_{\mathrm{out}} \widetilde{C}}]
    \end{equation}

    For each $c, c' \in C$
    \begin{align*}
        \braket{\Phi_{\mathrm{out}}}{c, c'} &= \frac{1}{\sqrt{d}} \sum_{k \in C} \braket{k}{c} \braket{k}{c'} \\
        &= \frac{1}{\sqrt{d}} \delta_{c, c'}
    \end{align*} Thereby 
    \begin{equation*}
        \braOPket{c, c'} {\Pi_{\mathrm{out}}} {c, c'} = \frac{1}{d} \delta_{c, c'}
    \end{equation*} which, if substituted into the definition of $\sigma_X$ in eq. (\ref{eq:sigmaXQuantum}) we get
    \begin{align*}
        \sigma_X &= \frac{1}{d} \sum_{a, c, c', x} p(a) P(x, c' \mid a, c) \frac{1}{d} \delta_{c, c'} \ketbra{x}{x} \\
        &= \frac{1}{d^2} \sum_{a, c, x} p(a) P(x, c \mid a, c) \ketbra{x}{x}\\
        &= \frac{1}{d^2} \sum_{a, x} p(a) \Tr_C(P)(x \mid a) \ketbra{x}{x} 
    \end{align*}

    The probability of obtaining the postselected Bell outcome is given by $p_\Phi = \tr(\sigma_X)$ which is just $$ p_\Phi = \frac{1}{d^2} \sum_{a \in A}p(a) \sum_{x \in X}\Tr_C(P)(x \mid a) $$ By Theorem \ref{thm:stoch-is-Tr}, $\tr (H_a) = 1$ for each $a$ if and only if $\Tr_C(P) \in \mathbf{Stoch}$; hence conditions 1 and 2 are equivalent. 

    If $\Tr_C(P)$ is stochastic then $\sum_{x}\Tr_C(P)(x \mid a) = 1$ for each $a$, and so $$ p_\Phi = \frac{1}{d^2} \sum_{a\in A} p(a) = \frac{1}{d^2} $$ and so 2 implies 3. 
    
    Conversely, suppose that $$ p_\Phi = \frac{1}{d^2} $$ for every probability distribution $p$ on $A$. Fix some $a_0 \in A$ and choose $p(a) = \delta_{a, a_0}$ then $$ \frac{1}{d^2} \sum_{x \in X} \Tr_C(P) (x \mid a_0) = \frac{1}{d^2} $$ and so $$\sum_{x \in X} \Tr_C(P)(x \mid a_0) = 1$$ Since $a_0$ was chosen arbitrarily and $\Tr_C(P)(x \mid a) \geq 0$ it follows that $\Tr_C(P)$ is stochastic and so 3 also implies 1 which by Theorem \ref{thm:stoch-is-Tr} implies 2, and so they are equivalent. 
\end{proof}

With normalised Bell state preparation and postselection, the success probability includes a dimensional factor $1/d^2$ which is generally dependent on the interaction and input. Theorem \ref{thm:quantumTFAE} identifies a class for which it equals $1/d^2$.

\begin{corollary}
   When all of the conditions in Theorem \ref{thm:quantumTFAE} hold, the normalised conditional state is
    \begin{align*}
        \rho^{\mathrm{cond}}_X &= \frac{\sigma_X}{p_\Phi} \\
        &= d^2\sigma_X \\
        &= \sum_{a \in A, x \in X} p(a) \Tr_C(P) (x \mid a) \ketbra{x}{x}
    \end{align*}
    Therefore the postselected Bell construction implements exactly the stochastic map obtained by taking the categorical trace of $P$ over the internal system $C$.
\end{corollary}


The construction above is a P-CTC in the sense of Lloyd et al.\ \cite{Lloyd2011-zc,Lloyd2011b} applied to the channel $\mathcal{E}_P$ rather than to a unitary. For a unitary with input $|\psi\rangle$ the same Bell-state preparation and postselection yield the unnormalised output $$\frac{1}{d}\,\mathrm{Tr}_C(U)\ket{\psi}$$ where the partial trace of the operator is given by $$\mathrm{Tr}_C(U) = \sum_{c \in C} \langle c|_{C_{\mathrm{out}}} U |c\rangle_{C_{\mathrm{in}}}$$ and success probability is then just $$\|\mathrm{Tr}_C(U)\,\ket{\psi}\|^2 / d^2$$ 

\begin{corollary}
    The following are equivalent: 
    \begin{center}
    $\Tr_C(P)$ is proportional to a stochastic map \\
    $\iff$ \\
    $p_\Phi$ is independent of $p$ \\
    $\iff$ \\
    $\tr (H_a)$ is constant in $a$.
    \end{center}
\end{corollary}

Note here that ``proportional'' permits a non-negative scalar multiplier. Let $\lambda$ be the common value of $\tr(H_a)$, then the probability of success is just $\lambda/d^2$. For $\lambda > 0$ the conditional stochastic map is $\lambda^{-1} \Tr_C(P)$. If $\lambda = 0$ then postselection never succeeds with no conditional output getting defined.

\section{Acknowledgements}

The author acknowledges the internet and all it does for sharing knowledge, and Prof. John Morton's $\mathrm{QUES^2T}$ Lab at University College London for kindly hosting Quantum Village and providing the space in which these results came from inspiration to fruition. This work was funded in part by the Mozilla Foundation, to whom Quantum Village is grateful for their belief in, and support of, our work.

\subsection{Use of AI}

GPT5.6-Sol and GPT6-Astra on High effort and Claude Fable-5.1 on Higher effort were used in the review process for this work. Their suggestions on mathematics were helpful, but many responses only led to confusion, as did their attempted rewrites of key surrounding sections. The text of this paper is entirely human written, with SPAG errors being identified through spellchecking software and with some expository passages revised following AI suggestions.

\bibliographystyle{plain}
\bibliography{main}

\end{document}